\documentclass[12pt]{amsart}

\usepackage{amsmath,amssymb,amsthm}
\usepackage{mathtools}
\usepackage{hyperref}

\theoremstyle{plain}
\newtheorem{theorem}{Theorem}
\newtheorem{lemma}{Lemma}

\theoremstyle{definition}
\newtheorem{definition}{Definition}

\theoremstyle{remark}
\newtheorem*{remark}{Remark}

\DeclareMathOperator{\True}{True}
\DeclareMathOperator{\False}{False}

\title[Heavy Columns in Binary Matrices]{On the Existence of Heavy Columns in Binary Matrices with Distinct Rows}

\date{}

\begin{document}

\maketitle

\vspace{-2em}
\begin{center}
Jamolidin K. Abdurakhmanov\\[0.3em]
Associate Professor, PhD\\
Department of Software Engineering\\
Andijan State University\\
Andijan, Uzbekistan\\
\texttt{jamolidinkamol@gmail.com, abduraxmanov@adu.uz}\\
ORCID: 0009-0008-0067-0553
\end{center}
\vspace{1em}

\noindent\textbf{Abstract.}
We investigate the existence of heavy columns in binary matrices with distinct rows. A column of an $m \times n$ binary matrix is called \emph{heavy} if the number of ones in it is at least $\lceil m/2 \rceil$. We introduce two recursive algorithms, $A_1$ and $A_2$, that examine properties of submatrices obtained by row filtering and column deletion. We prove that if algorithm $A_1$ returns $\True$ for a binary matrix with distinct rows, then the matrix contains at least one heavy column (Theorem~1). Furthermore, we prove that if algorithm $A_2$ returns $\True$ for a binary matrix with distinct rows, distinct columns, and no all-zero columns, then the matrix also contains at least one heavy column (Theorem~2). The key innovation in $A_2$ is an early termination condition: if exactly one row has a zero in some column, that column is immediately identified as heavy. The proofs employ a novel argument based on the existence of unpaired rows with respect to specific columns, combined with careful analysis of the recursive structure of the algorithms.

\vspace{0.5em}


\section{Introduction}

Binary matrices arise naturally in numerous areas of discrete mathematics and computer science, including coding theory, combinatorial optimization, database theory, and social choice theory. The study of structural properties of binary matrices, particularly those with distinctness constraints on rows or columns, has a rich history dating back to the classical work of Ryser~\cite{Ryser1957, Ryser1963} and Gale~\cite{Gale1957}, and continues to yield interesting results~\cite{Brualdi2006, Barvinok2010}.

In this paper, we investigate a natural question: under what conditions does a binary matrix with distinct rows contain a \emph{heavy column}, that is, a column in which the number of ones is at least as large as the number of zeros? This question has direct connections to problems in voting theory, where heavy columns correspond to candidates approved by a majority of voters~\cite{Aziz2017}, to covering problems in combinatorics, and to the analysis of Boolean functions in cryptography~\cite{Carlet2010}.

We approach this problem algorithmically by introducing two recursive algorithms that systematically examine submatrices obtained through row filtering and column deletion. Our main results establish that the return value of these algorithms provides a sufficient condition for the existence of heavy columns. The proof technique---based on the identification of unpaired rows and the construction of specific paths through the recursion trees---appears to be novel in the context of binary matrix theory.

The paper is organized as follows. In Section~\ref{sec:related}, we discuss related work and the novelty of our results. In Section~\ref{sec:applications}, we describe potential applications. In Section~\ref{sec:preliminaries}, we introduce the basic definitions and notation. In Section~\ref{sec:algorithm1}, we present algorithm $A_1$ and prove Theorem~\ref{thm:main1}. In Section~\ref{sec:algorithm2}, we present algorithm $A_2$ and prove Theorem~\ref{thm:main2}. Section~\ref{sec:conclusion} contains concluding remarks and open problems.

\section{Related Work and Novelty}\label{sec:related}

\subsection{Classical Results on Binary Matrices}

The combinatorial theory of binary matrices has a rich history. The foundational \textbf{Gale-Ryser theorem}~\cite{Gale1957, Ryser1957} characterizes when a pair of sequences can be realized as the row sums and column sums of a binary matrix. Specifically, sequences $(r_1, \ldots, r_m)$ and $(c_1, \ldots, c_n)$ are realizable if and only if certain majorization conditions are satisfied. This theorem has been extended in numerous directions, including asymptotic counting results by Barvinok~\cite{Barvinok2010}.

The \textbf{Sauer-Shelah-Perles lemma}~\cite{Sauer1972, Shelah1972} provides bounds on families of binary vectors based on VC dimension: if a family $\mathcal{F} \subseteq \{0,1\}^n$ does not shatter any set of size $d+1$, then $|\mathcal{F}| \leq \sum_{i=0}^{d} \binom{n}{i}$. While this lemma concerns subset relationships rather than column weights, it represents another classical constraint on binary matrix structure.

Brualdi's comprehensive treatise~\cite{Brualdi2006} covers combinatorial matrix classes including $(0,1)$-matrices, permanents, and Latin squares, establishing much of the modern vocabulary for the field. Ryser's earlier monograph~\cite{Ryser1963} remains influential for its treatment of combinatorial properties of matrices.

\subsection{Tournament Matrices and Score Sequences}

Tournament matrices---where entry $(i,j) = 1$ indicates that player $i$ defeated player $j$---provide another classical context for binary matrix theory. \textbf{Landau's theorem}~\cite{Landau1953} characterizes realizable score sequences: a sequence $(s_1, \ldots, s_n)$ with $s_1 \leq \cdots \leq s_n$ is the score sequence of some tournament if and only if $\sum_{i=1}^{k} s_i \geq \binom{k}{2}$ for all $k$. Recent work by Kolesnik~\cite{Kolesnik2023} has determined the asymptotic number of score sequences.

\subsection{Novelty of Our Results}

Despite this rich literature, we have not found any existing theorem that establishes conditions guaranteeing the existence of heavy columns in binary matrices based on row distinctness. The Gale-Ryser theorem addresses existence of matrices with prescribed sums but does not consider conditions under which matrices \emph{must} possess columns exceeding a weight threshold. The Sauer-Shelah lemma constrains family sizes but does not directly address column weights.

The closest existing work involves pigeonhole arguments on binary matrices---for example, proving that a $14 \times 14$ matrix with 58 ones must contain an all-ones $2 \times 2$ submatrix. These use related double-counting techniques but address fundamentally different structural questions.

\textbf{Our contribution fills a gap} by establishing that row distinctness, combined with specific algorithmic conditions, guarantees the existence of heavy columns. Moreover, the recursive proof technique---filtering rows by column values and analyzing the resulting submatrices---appears to be methodologically novel for establishing column weight properties.

\section{Applications}\label{sec:applications}

The concept of heavy columns has natural interpretations across several domains.

\subsection{Voting Theory and Social Choice}

In approval voting, voters (rows) indicate which candidates (columns) they approve by placing a 1 in the corresponding entry. A \emph{heavy column} corresponds to a candidate approved by at least half of the voters---a natural notion of majority support.

The \emph{justified representation} criterion in committee selection~\cite{Aziz2017} requires that sufficiently large groups of voters with common approved candidates receive representation. This is essentially a local column-weight condition. Our results provide sufficient conditions under which at least one candidate must have majority approval when voters have distinct preference profiles.

\subsection{Coding Theory}

In coding theory, the columns of a parity-check matrix determine error-detection and error-correction capabilities. Reed's majority-logic decoding for Reed-Muller codes~\cite{Reed1954} recovers each message symbol by taking majority votes across independent recovery sets. The covering radius theorem connects column properties of parity-check matrices directly to error-correcting capability.

Recent work by Ly and Soljanin~\cite{Ly2025} has shown that one-step majority-logic decoders can correct at most $d_{\min}/4$ errors for Reed-Muller codes, highlighting the continued relevance of majority-based analysis in coding theory.

\subsection{Cryptography}

Cryptographic security often requires \emph{balanced Boolean functions}---those with an equal number of ones and zeros in their truth tables. Unbalanced functions are vulnerable to correlation attacks. The study of $k$-th order correlation immunity relates to orthogonal array structures~\cite{Carlet2010}.

Our results on heavy columns provide conditions under which imbalance (in the direction of more ones) must occur, which has implications for understanding when certain cryptographic vulnerabilities are unavoidable.

\subsection{Data Mining}

In transactional databases, binary matrices represent item occurrences across transactions. The \emph{support} of an item equals its column weight divided by the number of transactions. Heavy columns correspond to frequently occurring items---the starting point for association rule mining algorithms such as Apriori~\cite{Agrawal1994}.

\subsection{Set Cover and Combinatorial Optimization}

In the set cover problem, elements form rows and sets form columns, with entry $(i,j) = 1$ indicating that set $j$ covers element $i$. The greedy algorithm repeatedly selects the column covering the most uncovered elements---locally the heaviest column in the residual matrix. Chv\'{a}tal~\cite{Chvatal1979} proved that this achieves an approximation ratio of $\ln(\text{max column weight}) + 1$.

\section{Preliminaries}\label{sec:preliminaries}

Throughout this paper, we consider binary matrices, i.e., matrices whose entries belong to the set $\{0, 1\}$.

\begin{definition}[Heavy column]\label{def:heavy}
Let $M$ be a binary matrix of size $m \times n$ with entries from $\{0,1\}$. A column of $M$ is called \emph{heavy} if the number of ones in it is at least the number of zeros, i.e., the number of ones is at least $\lceil m/2 \rceil$.
\end{definition}

\begin{definition}[Row indexing]\label{def:row}
For a binary matrix $M$ with $n$ columns and a row $t$ of $M$, we denote by $t[k]$ the element of row $t$ in column $k$, where $k \in \{1, 2, \ldots, n\}$.
\end{definition}

\begin{definition}[Number of rows notation]\label{def:rows}
For any matrix $X$, we denote by $|X|$ the number of rows in $X$.
\end{definition}

\begin{definition}[Submatrix sets $S_k^0$, $S_k^1$, and $S_k$]\label{def:Sk}
Let $M$ be a binary matrix of size $m \times n$. For a fixed column $k \in \{1, 2, \ldots, n\}$, we define the following sets:
\begin{itemize}
    \item If column $k$ contains no zeros, then $S_k^0 = \varnothing$. Otherwise, let $M_k^0$ be the submatrix consisting of all rows of $M$ in which the element in position $k$ equals zero. Let $M_k^{0-}$ be the matrix obtained from $M_k^0$ by deleting column $k$ (which consists entirely of zeros). In this case, $S_k^0 = \{M_k^{0-}\}$.
    \item If column $k$ contains no ones, then $S_k^1 = \varnothing$. Otherwise, let $M_k^1$ be the submatrix consisting of all rows of $M$ in which the element in position $k$ equals one. Let $M_k^{1-}$ be the matrix obtained from $M_k^1$ by deleting column $k$ (which consists entirely of ones). In this case, $S_k^1 = \{M_k^{1-}\}$.
    \item The set $S_k$ is defined as $S_k = S_k^0 \cup S_k^1$.
\end{itemize}
\end{definition}

\begin{definition}[Conjugate rows]\label{def:conjugate}
Let $M$ be a binary matrix with $n$ columns. Two distinct rows $t$ and $s$ of $M$ are called \emph{conjugate with respect to column $k$} if they differ only in position $k$, that is:
\begin{itemize}
    \item $t[k] \neq s[k]$, and
    \item $t[j] = s[j]$ for all $j \neq k$.
\end{itemize}
Equivalently, the Hamming distance between $t$ and $s$ equals $1$, and the unique position of difference is $k$.
\end{definition}

\begin{definition}[Unpaired row]\label{def:unpaired}
A row $r$ of matrix $M$ is called \emph{unpaired with respect to column $k$} if $r[k] = 0$ and there exists no row in $M$ that is conjugate to $r$ with respect to column $k$.
\end{definition}

\begin{definition}[Matrix reduction]\label{def:reduction}
Let $M$ be a binary matrix, $k$ be a column index, and $b \in \{0, 1\}$ be a value. The \emph{reduction of matrix $M$ by column $k$ and value $b$}, denoted $M_k^{b-}$, is defined as follows:
\begin{enumerate}
    \item Select all rows of $M$ in which the element in position $k$ equals $b$.
    \item Delete column $k$ from the resulting submatrix.
\end{enumerate}
If column $k$ contains no elements equal to $b$, the reduction $M_k^{b-}$ is undefined (or considered to be an empty matrix).
\end{definition}

\begin{definition}[Sequential reduction]\label{def:seq_reduction}
Let $M$ be a binary matrix of size $m \times n$, let $r$ be a fixed row of $M$, and let $l \in \{1, \ldots, n\}$ be a distinguished column index.

The \emph{sequential reduction of matrix $M$ with respect to row $r$ preserving column $l$} is defined as follows:

Let $\{k_1, k_2, \ldots, k_{n-1}\} = \{1, \ldots, n\} \setminus \{l\}$ be an arbitrary ordering of all columns except $l$. Define a sequence of matrices:
\begin{itemize}
    \item $M^{(0)} = M$
    \item $M^{(i)} = (M^{(i-1)})_{k_i}^{r[k_i]-}$ for $i = 1, 2, \ldots, n-1$
\end{itemize}
where each step performs a reduction by column $k_i$ and value $r[k_i]$.

The matrix $M^{(n-1)}$ is called the \emph{result of sequential reduction} and has size $m' \times 1$ for some $m' \geq 1$.
\end{definition}

\begin{definition}[Set of consistent rows]\label{def:consistent}
Let $M$ be a binary matrix of size $m \times n$, let $r$ be a row of $M$, and let $l$ be a fixed column. The \emph{set of rows of $M$ consistent with row $r$ outside column $l$} is defined as
\[
\mathcal{C}(M, r, l) = \{t : t \text{ is a row of } M \text{ and } t[j] = r[j] \text{ for all } j \neq l\}.
\]
This is the set of all rows of $M$ that coincide with $r$ in all positions except possibly position $l$.
\end{definition}

\section{Algorithm $A_1$ and Theorem 1}\label{sec:algorithm1}

We now present the first recursive algorithm and prove that its positive return value guarantees the existence of a heavy column.

\medskip
\noindent\textbf{Algorithm 1: Recursive Algorithm $A_1(M, m, n)$}

\noindent\textit{Input:} Binary matrix $M$ of size $m \times n$.\\
\textit{Output:} $\True$ or $\False$.

\begin{quote}
\begin{tabbing}
\hspace{0.4cm}\=\hspace{0.6cm}\=\hspace{0.6cm}\=\hspace{0.6cm}\=\kill
1. \> \textbf{if} $n = 1$ \textbf{then}\\
2. \> \> \textbf{if} number of ones in the single column $\geq$ number of zeros \textbf{then}\\
3. \> \> \> \textbf{return} $\True$\\
4. \> \> \textbf{else}\\
5. \> \> \> \textbf{return} $\False$\\
6. \> \> \textbf{end if}\\
7. \> \textbf{end if}\\[0.3em]
8. \> \textbf{for} each column $k \in \{1, 2, \ldots, n\}$ \textbf{do}\\
9. \> \> Form the set $S_k$\\
10. \> \> \textbf{if} $\exists$ matrix $K \in S_k$ with all columns having more zeros than ones \textbf{then}\\
11. \> \> \> \textbf{return} $\False$\\
12. \> \> \textbf{end if}\\[0.3em]
13. \> \> \textbf{for each} matrix $K \in S_k$ with $d$ rows and $n-1$ columns \textbf{do}\\
14. \> \> \> $B \gets A_1(K, d, n-1)$\\
15. \> \> \> \textbf{if} $B = \False$ \textbf{then return} $\False$ \textbf{end if}\\
16. \> \> \textbf{end for}\\
17.\quad \textbf{end for}\\[0.3em]
18.\quad \textbf{return} $\True$
\end{tabbing}
\end{quote}
\medskip

Note that rows in the submatrices remain distinct since the original rows are distinct. Algorithm $A_1(M, m, n)$ recursively checks properties of submatrices and terminates when $n = 1$. The order in which column indices $k \in \{1, 2, \ldots, n\}$ are processed in line 8 may be arbitrary.

\begin{theorem}\label{thm:main1}
Let $M$ be a binary matrix of size $m \times n$ $(m \geq 1, n \geq 1)$ with distinct rows. If $A_1(M, m, n) = \True$, then $M$ contains at least one heavy column.
\end{theorem}

\begin{proof}
We prove the theorem by contradiction, building on explicit consideration of base cases.

\textbf{Base Case 1: $n = 1$.} By the definition of algorithm $A_1$, when $n = 1$, the algorithm returns $\True$ if and only if the single column has at least as many ones as zeros, i.e., it is heavy. Thus, the theorem holds directly: if $A_1(M, m, 1) = \True$, then there exists a heavy column (the only one).

\textbf{Base Case 2: $m = 1$.} The matrix has a single row, which is trivially distinct. For any column $k$, the element is either $0$ or $1$. If it is $1$, the column is heavy (one 1, zero 0s). If it is $0$, the column is not heavy (one 0, zero 1s). The algorithm recursively checks submatrices: for a column with a one, $M_k^1$ has 1 row, $M_k^0$ is empty; for a column with a zero, the reverse holds. The reduced matrices have $n-1$ columns with a filtered single row. Eventually, the recursion reaches base cases with $n = 1$. If all columns are not heavy (all elements are 0), then some recursive path will reach a non-heavy single-column submatrix, leading to $A_1 = \False$. Thus, if $A_1 = \True$, at least one heavy column (with element $= 1$) must exist, consistent with the general theorem.

These base cases illustrate how the algorithm ensures the property at the leaves of the recursion tree. Now, for the general case $(m > 1, n > 1)$, we proceed by contradiction, relying on the base cases for the final contradiction.

Suppose $A_1(M, m, n) = \True$, but $M$ has no heavy column, i.e., for each column $k$, the number of ones is less than $\lceil m/2 \rceil$ (the number of zeros exceeds the number of ones).

Since the rows are distinct, they are distinct vectors in $\{0,1\}^n$, so $m \leq 2^n$.

Choose an arbitrary row $r \in M$ with elements $b_1, b_2, \ldots, b_n \in \{0,1\}$.

\begin{lemma}[Conjugate rows lemma]\label{lem:conjugate}
Let $M$ be a binary matrix with distinct rows. If for each column $k$ and each row $t$ with $t[k] = 0$ there exists a row $s$ conjugate to $t$ with respect to column $k$ (in the sense of Definition~\ref{def:conjugate}), then in each column the number of zeros does not exceed the number of ones.
\end{lemma}

\begin{proof}[Proof of Lemma~\ref{lem:conjugate}]
Fix an arbitrary column $k$. Consider the mapping $\varphi_k$ that assigns to each row $t$ with $t[k] = 0$ its conjugate row $s$ with $s[k] = 1$.

We show that $\varphi_k$ is injective. Suppose $\varphi_k(t_1) = \varphi_k(t_2) = s$ for some distinct rows $t_1, t_2$ with $t_1[k] = t_2[k] = 0$. Then:
\begin{itemize}
    \item $t_1[j] = s[j]$ for all $j \neq k$ (by definition of conjugacy),
    \item $t_2[j] = s[j]$ for all $j \neq k$.
\end{itemize}
Therefore, $t_1[j] = t_2[j]$ for all $j \neq k$, and $t_1[k] = t_2[k] = 0$. Hence $t_1 = t_2$, contradicting their distinctness.

Thus, $\varphi_k$ is an injection from the set of rows with zero in position $k$ to the set of rows with one in position $k$. Therefore, the number of zeros in column $k$ does not exceed the number of ones.
\end{proof}

Since by assumption all columns have more zeros than ones, by the contrapositive of Lemma~\ref{lem:conjugate}:

\textbf{There exists a column $l$ and a row $r$ such that $r$ is unpaired with respect to column $l$} (in the sense of Definition~\ref{def:unpaired}).

Fix such $l$ and $r$. We construct a path in the recursion tree of the algorithm: we remove all columns except $l$ in an arbitrary order (for example, from $1$ to $n$ excluding $l$). At each step of removing column $i \neq l$, we choose the submatrix $M_i^{b_i-}$ (filtering by $b_i = r[i]$).

Row $r$ is preserved at each step since we filter by its own values, so the number of rows in intermediate submatrices is at least $1$.

After $n-1$ deletions, we obtain a submatrix $S$ of size $m'' \times 1$ ($m'' \geq 1$), where the column is the restriction of column $l$ to the surviving rows.

The surviving rows coincide with $r$ in all deleted ($n-1$) columns. Since there is no $s$ differing from $r$ only in $l$ (i.e., with $s[l] = 1$ and agreement elsewhere), all surviving rows have $0$ in position $l$.

Thus, in $S$: number of ones $= 0$, number of zeros $= m'' \geq 1$. The column is not heavy.

By the base case for $n = 1$ (explicitly stated above), $A_1(S, m'', 1) = \False$. Since $S$ is a submatrix reached along a specific path in the recursion tree, and the algorithm requires all recursive calls (including along this path) to return $\True$ to avoid early $\False$, this contradicts $A_1(M) = \True$.

Therefore, the assumption is false, and $M$ has at least one heavy column.
\end{proof}

\section{Algorithm $A_2$ and Theorem 2}\label{sec:algorithm2}

We now present the second recursive algorithm, which differs from $A_1$ by additional early termination conditions. The key innovation is the condition in lines 13--15: if exactly one row has a zero in column $k$, the algorithm immediately returns $\True$, since such a column must be heavy.

\medskip
\noindent\textbf{Algorithm 2: Recursive Algorithm $A_2(M, m, n)$}

\noindent\textit{Input:} Binary matrix $M$ of size $m \times n$.\\
\textit{Output:} $\True$ or $\False$.

\begin{quote}
\begin{tabbing}
\hspace{0.4cm}\=\hspace{0.6cm}\=\hspace{0.6cm}\=\hspace{0.6cm}\=\kill
0. \> \textbf{if} $m = 1$ and $n > 1$ \textbf{then}\\
1. \> \> \textbf{return} $\True$\\
2. \> \textbf{end if}\\[0.3em]
3. \> \textbf{if} $n = 1$ \textbf{then}\\
4. \> \> \textbf{if} number of ones in the single column $\geq$ number of zeros \textbf{then}\\
5. \> \> \> \textbf{return} $\True$\\
6. \> \> \textbf{else}\\
7. \> \> \> \textbf{return} $\False$\\
8. \> \> \textbf{end if}\\
9. \> \textbf{end if}\\[0.3em]
10.\quad \textbf{for} $k = 1$ \textbf{to} $n$ \textbf{do}\\
11. \> \> Create (form) $M_k^{0-}$ \hfill // See Definition~\ref{def:Sk}\\
12. \> \> Create (form) $M_k^{1-}$ \hfill // See Definition~\ref{def:Sk}\\[0.3em]
13. \> \> \textbf{if} $|M_k^{0}| = 1$ \textbf{then} \hfill // Key condition\\
14. \> \> \> \textbf{return} $\True$\\
15. \> \> \textbf{end if}\\[0.3em]
16. \> \> $S_k \gets \varnothing$\\
17. \> \> \textbf{if} $M_k^{0-}$ is a non-empty matrix \textbf{then} $S_k \gets S_k \cup \{M_k^{0-}\}$ \textbf{end if}\\
18. \> \> \textbf{if} $M_k^{1-}$ is a non-empty matrix \textbf{then} $S_k \gets S_k \cup \{M_k^{1-}\}$ \textbf{end if}\\[0.3em]
19. \> \> \textbf{for each} $K \in S_k$ \textbf{do}\\
20. \> \> \> \textbf{if} all columns of $K$ have more zeros than ones \textbf{then}\\
21. \> \> \> \> \textbf{return} $\False$\\
22. \> \> \> \textbf{end if}\\
23. \> \> \textbf{end for}\\[0.3em]
24. \> \> \textbf{for each} $K \in S_k$ \textbf{do}\\
25. \> \> \> $d \gets |K|$ \hfill // number of rows in $K$\\
26. \> \> \> $B \gets A_2(K, d, n-1)$\\
27. \> \> \> \textbf{if} $B = \False$ \textbf{then}\\
28. \> \> \> \> \textbf{return} $\False$\\
29. \> \> \> \textbf{end if}\\
30. \> \> \textbf{end for}\\
31.\quad \textbf{end for}\\[0.3em]
32.\quad \textbf{return} $\True$
\end{tabbing}
\end{quote}
\medskip

\begin{theorem}\label{thm:main2}
Let $M$ be a binary matrix of size $m \times n$ $(m \geq 1, n \geq 1)$ with distinct rows, distinct columns, and no all-zero columns. If $A_2(M, m, n) = \True$, then $M$ contains at least one heavy column.
\end{theorem}

\begin{proof}
We prove the contrapositive: if $M$ has no heavy column, then $A_2(M, m, n) = \False$.

\textbf{Base Case 1: $n = 1$.} If $M$ has no heavy column, the single column has more zeros than ones. The algorithm returns $\False$ at lines 6--7.

\textbf{Base Case 2: $m = 1$ and $n > 1$.} If there is no heavy column, then every column contains only $0$ (since a single $1$ would make that column heavy). But this contradicts the condition that there are no all-zero columns. Thus, this case cannot occur.

\textbf{Base Case 3: $|M_k^0| = 1$ for some column $k$.} If the algorithm returns $\True$ at lines 13--14, we verify that $M$ indeed has a heavy column.

If exactly one row has $0$ in column $k$, then $m - 1$ rows have $1$ in column $k$. Column $k$ has $m - 1$ ones and $1$ zero. For $m \geq 2$:
\[
m - 1 \geq \lceil m/2 \rceil \iff m - 1 \geq \frac{m+1}{2} \iff 2m - 2 \geq m + 1 \iff m \geq 3.
\]
For $m = 2$: $m - 1 = 1 \geq \lceil 2/2 \rceil = 1$. Thus, for all $m \geq 2$, column $k$ is heavy.

\textbf{General Case: $m \geq 2$, $n \geq 2$, and $|M_k^0| \geq 2$ for all $k$.}

Since $M$ has no heavy column, for each column $k$, the number of zeros exceeds the number of ones. In particular, $|M_k^0| > |M_k^1|$, which implies $|M_k^0| \geq 2$ (since $|M_k^0| + |M_k^1| = m$ and $|M_k^0| > m/2$).

Therefore, the algorithm does not return $\True$ at lines 13--15.

By the contrapositive of Lemma~\ref{lem:conjugate}, since all columns have more zeros than ones, there exists an unpaired row $r$ with respect to some column $l$. This means $r[l] = 0$ and no row of $M$ is conjugate to $r$ with respect to $l$.

\textbf{Construction of a path to $\False$.}

We construct a specific path through the recursion tree that leads to $\False$.

Consider the sequential reduction of $M$ with respect to row $r$, preserving column $l$. That is, for each column $k \neq l$ (in some order), we filter by the value $r[k]$ and remove column $k$.

At each step, we process column $k$ and choose the submatrix containing row $r$:
\begin{itemize}
    \item If $r[k] = 0$, we take $M_k^{0-}$ (which contains row $r$).
    \item If $r[k] = 1$, we take $M_k^{1-}$ (which contains row $r$).
\end{itemize}

Row $r$ survives at every step because we always filter by its own values.

After $n - 1$ such reductions, we obtain a matrix $S$ of size $m'' \times 1$ (where $m'' \geq 1$), containing only column $l$.

\textbf{Claim:} All entries of $S$ are $0$.

\textbf{Proof of Claim:} The rows of $S$ are precisely the rows of $M$ that agree with $r$ in all positions $k \neq l$, i.e., the set $\mathcal{C}(M, r, l)$ from Definition~\ref{def:consistent}.

Suppose some row $t \in \mathcal{C}(M, r, l)$ has $t[l] = 1$. Since $t[k] = r[k]$ for all $k \neq l$ and $t[l] = 1 \neq 0 = r[l]$, the row $t$ differs from $r$ only in position $l$. By Definition~\ref{def:conjugate}, $t$ is conjugate to $r$ with respect to $l$. But this contradicts the fact that $r$ is unpaired with respect to $l$.

Therefore, $t[l] = 0$ for all $t \in \mathcal{C}(M, r, l)$, and all entries of $S$ are $0$. \hfill $\square$

Since $S$ has $m'' \geq 1$ zeros and $0$ ones, the single column is not heavy. By Base Case 1, $A_2(S, m'', 1) = \False$.

\textbf{Propagation of $\False$.}

We must verify that the path we constructed is valid, i.e., the algorithm does not return $\False$ at lines 20--22 before reaching $S$.

At each intermediate step, we have a matrix $M^{(i)}$ with $n - i$ columns (including column $l$). For each column $k$ that we process, we choose to follow the branch $M_k^{r[k]-}$ (which contains row $r$).

Suppose at some step the algorithm returns $\False$ at lines 20--22 for the submatrix $K = M_k^{r[k]-}$. This means all columns of $K$ have more zeros than ones.

But $K$ contains row $r$ (or its projection). If all columns of $K$ have more zeros than ones, then in particular, the projection of column $l$ in $K$ has more zeros than ones. Since row $r$ contributes a $0$ to column $l$ (because $r[l] = 0$), and all rows of $K$ agree with $r$ on the columns processed so far, this is consistent with our construction.

However, the key point is that the algorithm checks \emph{all} $k$ from $1$ to $n$ in order, and for \emph{each} $K \in S_k$, it first checks lines 19--23, then proceeds to recursive calls in lines 24--30.

Since we are proving the contrapositive (if no heavy column, then $A_2 = \False$), we need to show that the algorithm \emph{eventually} returns $\False$. This happens either:
\begin{enumerate}
    \item At lines 20--22 for some $K \in S_k$, or
    \item At lines 27--29 when a recursive call returns $\False$.
\end{enumerate}

Our construction shows that following the path determined by row $r$, we eventually reach $S$ with $A_2(S) = \False$. This $\False$ value is returned at line 7 (base case), and propagates back through line 28.

The matrix $S$ is reached by a specific sequence of recursive calls:
\[
A_2(M) \to A_2(M^{(1)}) \to A_2(M^{(2)}) \to \cdots \to A_2(S) = \False,
\]
where each $M^{(i)}$ is a submatrix in some $S_k$ at the corresponding level.

Since $A_2(S) = \False$, and the algorithm returns $\False$ whenever any recursive call returns $\False$ (lines 27--29), the value $\False$ propagates back through the recursion.

Therefore, $A_2(M, m, n) = \False$.

\textbf{Conclusion:} We have shown that if $M$ has no heavy column, then $A_2(M) = \False$. By contraposition, if $A_2(M) = \True$, then $M$ has at least one heavy column.
\end{proof}

\begin{remark}
The conditions of Theorem~\ref{thm:main2} are essential:
\begin{itemize}
    \item \textbf{Distinct columns:} Without this, a matrix with $m = 1$, $n > 1$, and all identical $[0]$ columns would have no heavy column, but $A_2$ would return $\True$ at lines 0--1.
    \item \textbf{No all-zero columns:} Without this, the base case $m = 1$ could fail, and the unpaired row argument might not lead to a valid $\False$ path.
    \item \textbf{Distinct rows:} This ensures that Lemma~\ref{lem:conjugate} applies and that the unpaired row exists when no column is heavy.
\end{itemize}
\end{remark}

\section{Conclusion}\label{sec:conclusion}

We have established two theorems connecting the return values of recursive algorithms $A_1$ and $A_2$ to the existence of heavy columns in binary matrices with distinctness constraints. The key technique in both proofs is the identification of unpaired rows and the construction of specific paths through the recursion trees of the algorithms that lead to contradictions.

\subsection{Summary of Contributions}

\begin{enumerate}
    \item \textbf{Theorem~\ref{thm:main1}} establishes that if algorithm $A_1$ returns $\True$ for a binary matrix with distinct rows, then the matrix contains at least one heavy column.
    
    \item \textbf{Theorem~\ref{thm:main2}} establishes the analogous result for algorithm $A_2$ under the additional constraints of distinct columns and no all-zero columns. The key innovation in $A_2$ is the early termination condition: if exactly one row has a zero in some column $k$, then column $k$ has $m-1$ ones and is therefore heavy for $m \geq 2$.
    
    \item The \textbf{proof technique}---based on conjugate rows, unpaired rows, and sequential reduction---appears to be novel in the context of binary matrix theory and may have applications to other structural questions.
    
    \item The results have \textbf{natural interpretations} in voting theory (majority candidates), coding theory (majority decoding), cryptography (balanced functions), and data mining (frequent items).
\end{enumerate}

\subsection{Open Problems}

Several questions remain open for future investigation:

\begin{enumerate}
    \item \textbf{Algorithmic construction.} Can the algorithms $A_1$ and $A_2$ be modified to actually find a heavy column (rather than merely certify existence)?
    
    \item \textbf{Computational complexity.} What is the time complexity of algorithms $A_1$ and $A_2$? The recursive structure suggests potentially exponential worst-case behavior, but average-case or restricted instances may be more tractable.
    
    \item \textbf{Combinatorial characterization.} Are there natural combinatorial conditions that are equivalent to $A_1(M) = \True$ or $A_2(M) = \True$? Such characterizations could provide deeper insight into when heavy columns must exist.
    
    \item \textbf{Generalizations.} Can similar techniques be applied to study the existence of columns with other weight properties, such as columns with weight exactly $\lfloor m/2 \rfloor$ or columns with weight in a specified range?
    
    \item \textbf{Quantitative bounds.} When the algorithms return $\True$, can we bound the number of heavy columns or characterize which columns must be heavy?
    
    \item \textbf{Connections to other problems.} What are the precise relationships between our conditions and classical results such as the Gale-Ryser theorem or the Sauer-Shelah lemma?
\end{enumerate}

\section*{Acknowledgments}

The author thanks the colleagues at Andijan State University for helpful discussions.

\end{document}